\documentclass[runningheads,letterpaper]{llncs}

\usepackage{svninfo} 

\usepackage{amsmath,amssymb}

\newcommand{\NP}{\mathrm{NP}}
\newcommand{\APX}{\mathrm{APX}}
\newcommand{\Zset}{\mathbb{Z}}
\newcommand{\card}[1]{\left| #1 \right|}
\newcommand{\cT}{\mathcal{T}}
\newcommand{\cH}{\mathcal{H}}

\newcommand{\seg}[2]{\left[ #1, #2 \right]} 
\newcommand{\inda}[1]{\left(#1_a \right)_{a\in A}}

\newcommand{\pbmsu}{\textsc{MinSU}}
\newcommand{\pbsul}{\textsc{Soapy Union}}
\newcommand{\pbsu}{\textup{SU}}
\newcommand{\pbmaux}{\textsc{MinAux}}
\newcommand{\pbaux}{\textsc{Aux}}
\newcommand{\pbmvc}{\textsc{MinVC}}
\newcommand{\pbvc}{\textup{VC}}
\newcommand{\pbvcl}{\textsc{Vertex Cover}}

\newcommand\PNP{\ensuremath{\mathrm{P}=\NP}}

\newcommand{\defpb}[4]{
\begin{trivlist}
 \item[]\emph{Name}: #1
 \item[]\emph{Input}: #2
 \item[]\emph{Solution}: #3
 \item[]\emph{Measure}: #4
\end{trivlist}
}

\begin{document}

\svnInfo $Id: ms_complex.tex 23121 2012-04-20 19:27:11Z toto $

\title{Minimum soapy union} 

\titlerunning{Minimum soapy union}

\author{Francois Nicolas\inst{1} \and Sebastian
  B\"ocker\inst{1}}

\institute{Lehrstuhl f{\"u}r Bioinformatik,
  Friedrich-Schiller-Universit{\"a}t Jena, Ernst-Abbe-Platz 2, Jena, Germany,
  \email{\{francois.nicolas, sebastian.boecker\}@uni-jena.de}
}

\authorrunning{F.~Nicolas, S.~B\"ocker}

\maketitle

\sloppy


\begin{abstract} 
  Define \textsc{Minimum} \pbsul{} (\pbmsu) as the following optimization problem:
  given a $k$-tuple $(X_1, X_2, \dotsc, X_k)$ of finite integer sets, find a
  $k$-tuple $(t_1, t_2, \dotsc, t_k)$ of integers that minimizes the
  cardinality of ${(X_1 + t_1)} \cup {(X_2 + t_2)} \cup \dotsb \cup {(X_n +
    t_k)}$.  We show that \pbmsu{} is $\NP$-complete,
  $\APX$-hard, and polynomial for fixed~$k$.

  \pbmsu{} appears naturally in the context of
  protein shotgun sequencing: Here, the protein is cleaved into short and
  overlapping peptides, which are then analyzed by tandem mass spectrometry.
  To improve the quality of such spectra, one then asks for the mass of the
  unknown prefix (the~shift) of the spectrum, such that the resulting shifted
  spectra show a maximum agreement.  For real-world data the problem is even
  more complicated than our definition of \pbmsu; but our
  intractability results clearly indicate that it is unlikely to find a
  polynomial time algorithm for shotgun protein sequencing.
%
\end{abstract}


\section{Introduction}

The aim of this paper is to study the computational complexity of the
following optimization problem:

\defpb{\textsc{Minimum} \pbsul{} (\pbmsu)}
{a finite set $A$ and an indexed family $\inda{X}$ of non-empty finite sets of rational integers.}
{an indexed family $\inda{t}$ of rational integers.}
{the cardinality of $\bigcup_{a \in A} (X_a + t_a)$.}

Let us name \pbsul{} (\pbsu) the decision problem associated with \pbmsu.
The names have been chosen by analogy with the \textsc{Soapy Set Cover} problem \cite{nicolas08hardness}.
Clearly, \pbsu{} is a \emph{number problem} \cite{garey79computers}.
\pbmsu{} can be seen as a generalization of the \textsc{Subset Matching} problem \cite{clifford07approximate}:
optimally solving \textsc{Subset Matching} is equivalent to optimally solving the restriction of \pbmsu{} to those instances $\inda{X}$  such that the cardinality of $A$ equals~$2$.

\pbmsu{} naturally appears in the context of
protein shotgun sequencing~\cite{bandeira04shotgun,bandeira08automated,bandeira08de-novo}.  (This problem must not be confused with the more
widely known peptide shotgun sequencing.)  Sequencing the protein means that
we want to determine its amino acid sequence.  We assume that no genomic
information is available for the protein, so that its sequence cannot be
derived from the genomic information.  This is the case for many proteins
even in humans, monoclonal antibodies being an important
example~\cite{bandeira08automated}.  Experimentally, the protein is cleaved
into short and overlapping peptides, which are then analyzed by tandem mass
spectrometry.  To improve the quality of such spectra, one then asks for the
mass of the unknown prefix (the~shift) of the spectrum, such that the
resulting shifted spectra show a maximum agreement.  For real-world data the
problem is even more complicated than our definition of \pbmsu;
but our intractability results clearly indicate that it is unlikely to find a
polynomial time algorithm for shotgun protein sequencing.

\paragraph{Contribution.} 

In Section~\ref{sec:msu-easy}, 
we prove 
that  \pbsu{} belongs to $\NP$ and 
that \pbmsu{} can be solved in polynomial time for fixed~$A$.
In Section~\ref{sec:msu-hard}, we show that \pbsu{} is strongly $\NP$-hard; 
furthermore, 
we prove that there exists a real number $\rho > 1$ such that if \pbmsu{} is $\rho$-approximable in pseudo-polynomial time then $\PNP$. 

\paragraph{Notation and definitions.} 

For every finite set $S$, $\card{S}$ denotes the cardinality of~$S$.
For all sets $A$ and $S$, $S^A$ denotes the set of all families  of elements of $S$ indexed by~$A$.

The ring of rational integers is denoted~$\Zset$.
For every integer $n \ge 0$,
$\seg{1}{n}$ denotes the set of all $k \in \Zset$ such that $1 \le k \le n$.

A(n undirected) \emph{graph} is a pair $G = (V, E)$, where $V$ is a finite set and $E$ is a set of $2$-element subsets of~$V$: 
the elements of $V$ are the \emph{vertices} of $G$,
the elements of $E$ are the \emph{edges} of $G$,
and for each edge $e \in E$, the elements of $e$ are the \emph{extremities} of~$e$.

Let \textsc{Min} be a minimization problem.
The \emph{decision problem associated with} \textsc{Min} is: 
given an instance $I$ of \textsc{Min} and an integer $k \ge 0$,  
decide whether there exists a solution of \textsc{Min} on $I$ with measure at most~$k$.


\section{Membership} \label{sec:msu-easy}

For each instance $\inda{X}$ of \pbmsu, 
the set of all feasible solutions  of \pbmsu{} on $\inda{X}$ equals $\Zset^A$, which is infinite.
Therefore,  
\pbmsu{} is not an \emph{$\NP$-optimization problem} \cite{ausiello03complexity}, 
and thus the membership of \pbsu{} in $\NP$ is not completely trivial.

Let $G = (V, E)$ be a graph.
A \emph{disconnection} of $G$ is 
a pair $(B, C)$ such that 
$B \ne \emptyset$, 
$C \ne \emptyset$,  
$B \cap C = \emptyset$,
$V = B \cup C$,
and 
for every $(b, c) \in B \times C$, $\{ b, c \} \notin E$.
A graph is called \emph{disconnected} if it admits a disconnection.
A graph that is not disconnected is called \emph{connected}.

Let $\inda{Y}$ be an indexed family of sets.
The \emph{intersection graph} of $\inda{Y}$ is defined as follows: 
its vertex set equals $A$ and 
for all $b$, $c \in A$ with $b \ne c$, 
$\{ b, c \}$  is one of its edges if, and only if, $Y_b \cap Y_c \ne \emptyset$.

\begin{lemma} \label{lem:discon-opt}
Let $\inda{Y}$ be an instance of \pbmsu.
If the intersection graph of $\inda{Y}$ is disconnected 
then 
there exists $\inda{u} \in \Zset^A$ such that
\begin{equation} \label{eq:X+t=A}
\card{\bigcup_{a \in A}  (Y_a + u_a )}
< 
\card{\bigcup_{a \in A}  Y_a}   \, . 
\end{equation}
\end{lemma}

\begin{proof}
For each subset $B \subseteq A$, put $Y_{B} = \bigcup_{b \in B} Y_b$.
Let $(B, C)$ be a disconnection of the intersection graph of $\inda{Y}$.
Let 
$r \in Y_B$ 
and 
$s \in Y_C$
be fixed.
Set 
$u_b = - r$ for every~$b \in B$ 
and 
$u_c = - s$ for every~$c \in C$.
On the one hand, we have 
$Y_B \cap Y_C = \emptyset$
and 
$Y_A = Y_B \cup Y_C$,
so 
\begin{equation} \label{eq:A=BC}
\card{Y_A }  = \card{Y_B} + \card{Y_C} \, .
\end{equation}
On the other hand, we have 
$$
\bigcup_{a \in A}  (Y_a + u_a ) = (Y_B - r) \cup (Y_C - s) 
$$
and 
$$
(Y_B - r) \cap (Y_C - s) \ne \emptyset 
$$
because $0 \in (Y_B - r) \cap (Y_C - s)$; 
it follows 
\begin{equation} \label{eq:A+t=BC}
 \card{\bigcup_{a \in A}  (Y_a + u_a ) }
 < 
  \card{Y_B}
  +
  \card{Y_C} \, .
\end{equation}
It now suffices to combine Equations~\eqref{eq:A=BC} and \eqref{eq:A+t=BC} to obtain Equation~\eqref{eq:X+t=A}.
\qed
\end{proof}

Lemma~\ref{lem:discon-opt} can be restated as follows:

\begin{lemma} \label{lem:opt-con}
Let $\inda{X}$ be an instance of \pbmsu.
For any optimum solution $\inda{t}$ of \pbmsu{} on  $\inda{X}$,
the intersection graph of $\left( X_a + t_a \right)_{a \in A}$ is connected.
\end{lemma}

\begin{proof}
Let $\inda{t} \in \Zset^A$ be such that the intersection graph of $\left( X_a + t_a \right)_{a \in A}$  is disconnected.
Set $Y_a = X_a + t_a$ for each $a \in A$.
By Lemma~\ref{lem:discon-opt}, 
there exists $\left( u_a \right)_{a \in A} \in \Zset^A$ such that  Equation~\eqref{eq:X+t=A} holds.
It follows that $\left( t_a + u_a \right)_{a \in A}$ is a better solution of \pbmsu{} on $\inda{X}$ than $\inda{t}$.
\qed
\end{proof}

\begin{definition} \label{def:SGpi}
Let $G = (V, E)$ be a graph.
Put $$\tilde E = \left\{ (a, b) \in V \times V : \{ a, b \} \in E \right\}\, .$$
An \emph{antisymmetric edge-weight function} on $G$ is 
a function $\varpi$ from $\tilde E$ to $\Zset$ such that 
$\varpi(b, c) = - \varpi(c, b)$
for every $(b, c) \in \tilde E$.
For each antisymmetric edge-weight function $\varpi$ on $G$, 
define $S(G, \varpi)$ as the set of all $\left( t_a \right)_{a \in V} \in \Zset^V$ such that 
$t_b - t_c = \varpi(b, c)$ 
for all $(b, c) \in \tilde E$.
\end{definition}

Let us comment Definition~\ref{def:SGpi}.
The function $\varpi$ assigns both a magnitude and an orientation to each edge of~$G$:
for all $a$, $b \in V$ such that $\{ a, b \} \in E$, 
the magnitude of $\{ a, b \}$ is the absolute value of $\varpi(a, b)$
and 
the orientation of $\{ a, b \}$ 
is determined by the sign of $\varpi(a, b)$.
It is clear that for every $\left( t_a \right)_{a \in V} \in S(G, \varpi)$ and every $u \in \Zset$, 
$ \left( t_a + u \right)_{a \in V} \in S(G, \varpi)$.
If $G$ is connected then 
either $S(G, \varpi)$ is empty or there exists $\left( t_a \right)_{a \in V} \in \Zset^V$ such that 
$S(G, \varpi) = \left\{ \left( t_a + u \right)_{a \in V} : u \in \Zset \right\}$.
If $G$ is connected and $S(G, \varpi) \ne \emptyset$
then 
for any $(b, u) \in V \times \Zset$,
the unique element $\left( t_a \right)_{a \in V} \in S(G, \varpi)$ that satisfies $t_b = u$ is computable from $G$, $\varpi$, $b$, and $u$ in polynomial time.
A \emph{closed walk} in $G$ is a finite sequence 
$(a_0, a_1, a_2, \dotsc, a_k)$ 
such that 
$a_0 = a_k$
and 
$\{ a_{i - 1}, a_i \} \in E$ for every $i \in \seg{1}{k}$;
the \emph{weight} of $(a_0, a_1, a_2, \dotsc, a_k)$ under $\varpi$ is defined as 
$\varpi(a_0, a_1) +  \varpi(a_1, a_2) + \dotsb +   \varpi(a_{k - 1}, a_k)$.
A (simple) \emph{cycle} in $G$ is a closed walk  $(a_0, a_1, a_2, \dotsc, a_k)$  in $G$ such that 
for all $i$, $j \in \seg{1}{k}$, 
$a_i = a_j$ implies $i = j$. 
The following three conditions are equivalent:
\begin{enumerate}
 \item The set $S(G, \varpi)$ is non-empty. 
 \item The weight under $\varpi$ of every closed walk in $G$ equals~$0$.
\item The weight under $\varpi$ of every cycle in $G$ equals~$0$.
\end{enumerate}
The second and third conditions can be thought as abstract forms of Kirchhoff's voltage law.

A  \emph{tree} is a connected graph with one fewer edges than vertices, or equivalently, an acyclic connected graph.
An arbitrary graph $G = (V, E)$ is connected 
if, and only if, 
there exists a subset $E' \subseteq E$ such that $(V, E')$ is a tree 
($(V, E')$ is then called a \emph{spanning tree} of~$G$). 

\begin{lemma} \label{lem:tree-Xa}
Let $\inda{X}$ be an instance of \pbmsu. 
There exist a tree $H$ with vertex set $A$ and an antisymmetric edge-weight function $\varpi$ on $H$ that satisfy the following two conditions:
\begin{enumerate}
 \item \label{cond:range-varpi}
Every integer in the range of $\varpi$ can be written as the difference of two elements of $\bigcup_{a \in A} X_a$.
 \item \label{cond:SUpi-opt}
Every element of $S(H, \varpi)$ is an optimum solution of \pbmsu{} on  $\inda{X}$.
\end{enumerate}
\end{lemma}

\begin{proof}
Let $\inda{t}$ be an optimum solution of \pbmsu{} on $\inda{X}$.
Let $H$ be a spanning tree of the intersection graph of $\left( X_a + t_a \right)_{a \in A}$: such a tree exists 
by Lemma~\ref{lem:opt-con}.
Let  $\varpi$ be the antisymmetric edge-weight function on $H$ defined by:
for all $b$, $c \in A$ such that $\{ b, c \}$ is an edge of $H$,  
$\varpi(b, c) = t_b - t_c$. 

For all $b$, $c \in A$, 
 such that $\{ b, c \}$ is an edge of the intersection graph of $\left( X_a + t_a \right)_{a \in A}$,
$(X_b + t_b) \cap (X_c + t_c)$ is non-empty, 
and thus 
$t_b - t_c$ belongs to $X_c - X_b$.
Therefore, the first condition holds.
Now, remark that  $S(H, \varpi) = \left\{ \left( t_a + u \right)_{a \in A} : u \in \Zset \right\}$,
so the second condition holds.
\qed
\end{proof}

\begin{theorem} \label{th:SU-inNP}
\pbsu{} belongs to $\NP$.
\end{theorem}

\begin{proof}
Let $\left(\inda{X}, k \right)$ be an arbitrary instance of \pbsu.
We propose the following (non-deterministic) algorithm to decide whether $\left(\inda{X}, k \right)$ is a yes-instance of \pbsu:
\begin{itemize}
 \item 
Guess a tree $H$ with vertex set $A$ and an antisymmetric edge-weight function $\varpi$ on $H$ such that the first condition of Lemma~\ref{lem:tree-Xa} holds.
\item 
Compute an element $\inda{t} \in S(H, \varpi)$.
\item 
Check whether the cardinality of $\bigcup_{a \in A} (X_a + t_a)$ is at most~$k$.
\end{itemize}
By Lemma~\ref{lem:tree-Xa}, the algorithm is correct.
Moreover, 
the bit-length of the guess (\emph{i.e}, the ordered pair $(H, \varpi)$) 
is polynomial in 
the bit-length of the input (\emph{i.e}, the instance $\inda{X}$), 
so the algorithm can be implemented in non-deterministic polynomial time.
\qed
\end{proof}

Let $m$ be a positive integer and let $X$ be a subset of $\Zset$ such that $X = -X$.
On each given $m$-edge graph,
there are exactly $\card{X}^m$ distinct antisymmetric edge-weight functions whose ranges are subsets of~$X$. 

Let $n$ be a positive integer and let $\cT_n$ denote the set of all trees with vertex set $\seg{1}{n}$.
\emph{Cayley's formula} ensures $\card{\cT_n} = n^{n - 2}$ \cite{harary73graphical}. 
Moreover, every tree can be reconstructed in polynomial time from its \emph{Pr\"ufer code} \cite{harary73graphical},
so $\cT_n$ is enumerable in $O \left( n^{O(n)} \right)$ time.

\begin{theorem}
There exists an algorithm  that, 
for each instance $\inda{X}$ of \pbmsu{} given as input,
returns an optimum solution of \pbmsu{} on $\inda{X}$ in 
$O \left( N^{O(\card{A})} \right)$ 
time, where $N$ denotes the bit-length of $\inda{X}$.
\end{theorem}

\begin{proof}
Put $U = \bigcup_{a \in A} X_a$.
Let $\cH$ denote the set of all ordered pairs of the form $(H, \varpi)$, 
where 
$H$ is a tree with vertex set $A$ 
and 
$\varpi$ is an antisymmetric edge-weight function on $H$ whose range is a subset of $U - U$.
We propose the following algorithm to solve \pbmsu{} on $\inda{X}$:
\begin{itemize}
 \item For each $(H, \varpi) \in \cH$, compute an element of $S(H, \varpi)$.
 \item Return a best solution of \pbmsu{} on $\inda{X}$ among those computed at the previous step.
\end{itemize}
By Lemma~\ref{lem:tree-Xa}, the algorithm returns an optimum solution of \pbmsu{} on $\inda{X}$.
Moreover, remark that $\card{\mathcal{H}} = \card{A}^{\card{A} - 2} \card{U - U}^{\card{A} - 1}$
and that $\cH$ is enumerable in $O \left( N^{O(\card{A})} \right)$ time.
Therefore, the algorithm can be implemented to run in $O \left( N^{O(\card{A})} \right)$ time.
\qed
\end{proof}


\section{Hardness} \label{sec:msu-hard}

The aim of this section is prove the hardness results for \pbmsu.

Let $G = (V, E)$ be a graph. 
A \emph{vertex cover} of $G$ is a subset $C \subseteq V$ such that $C \cap e \ne \emptyset$ for every $e \in E$: 
a vertex cover is a subset of vertices that contains at least one extremity of each edge.

\defpb
{\textsc{Minimum} \pbvcl{} (\pbmvc)}
{a graph~$G$.}
{a vertex cover $C$ of~$G$.}
{the cardinality of~$C$.}
The decision problem associated with \pbmvc{} is named \pbvcl{} (\pbvc).
It is well-known that \pbvc{} is $\NP$-complete \cite{garey79computers}.

To prove that \pbsu{} is (strongly) $\NP$-complete, 
we show that \pbvc{} Karp-reduces to (a suitable restriction of) \pbsu.
The following gadget plays a crucial role in our reduction as well as in other reductions 
that can be found in the literature \cite{nicolas08hardness,michael10complexity}:

\begin{definition}
For each integer $n \ge 1$, 
define $R_n = \left\{ (i - 1) n^2 + i^2 : i \in \seg{1}{n} \right\}$.
\end{definition}
              
A \emph{Golomb ruler} \cite{gardner83wheels,sidon32satz,babcock53intermodulation} is a finite subset $R \subseteq \Zset$ that satisfies the following three equivalent conditions:
\begin{itemize}
 \item For every $t \in \Zset$,  $t \ne 0$ implies  $\card{R \cap (R + t)} \le 1$.
 \item For every integer $d > 0$, there exists at most one $(r, s) \in R \times R$ such that $r - s = d$.
 \item For all $r_1$, $r_2$, $s_1$, $s_2 \in R$, $r_1 + r_2 = s_1 + s_2$ implies $\{ r_1, r_2 \} = \{ s_1, s_2 \}$.
\end{itemize}
Actually, only the first condition is referred to in what follows.
Among other convenient properties our gadget sets are Golomb rulers:

\begin{lemma} \label{lem:golomb}
Let $n$ be a positive integer.
The following four properties hold.
\begin{enumerate}
\item \label{ppty:subset} The least element of $R_n$ is $1$ and the greatest element of $R_n$ is~$n^3$.
\item  \label{ppty:card} The cardinality of $R_n$ equals~$n$.
\item \label{ppty:distn} 
The distance between any two elements of $R_n$ is at least $n^2 + 3$.
\item \label{ppty:Golomb} $R_n$ is a Golomb ruler.
\end{enumerate}
\end{lemma}

\begin{proof}
Properties~\ref{ppty:subset} and~\ref{ppty:card} are clear.
Proofs of Property~\ref{ppty:Golomb} can be found in \cite{nicolas08hardness,michael10complexity}.
Finally, 
remark that for every $i \ge 1$, 
we have  
$$
\left( i n^2 + {(i + 1)}^2 \right)
-
 \left( (i - 1) n^2 + i^2 \right) 
  =
 n^2 + 2i + 1 
 \ge 
 n^2 + 3 \, .
$$
Hence, the distance the distance between any two consecutive elements of $R_n$ is at least 
$n^2 + 3$, and thus Property~\ref{ppty:distn} holds.
\qed
\end{proof}

\begin{theorem} \label{th:MSU-NPC}
\pbsu{} is strongly $\NP$-hard. 
\end{theorem}

\begin{proof}
Put  $f(x) = \left( \tfrac{1}{4} x + 2 \right)^3 + \tfrac{1}{2}x - 4$.
Let \pbaux{} denote the restriction of \pbsu{} to those instances $\left( \inda{X}, k \right)$ 
such that 
the absolute value of every integer in 
$\left\{ k \right\} \cup \bigcup_{a \in  A} X_a$ is at most $f\left( \max_{a \in A} \card{X_a} \right)$.
We prove that  \pbaux{} is $\NP$-hard which implies the theorem.
More precisely, we show that \pbvc{}  Karp-reduces to \pbaux.

\paragraph{Presentation of the reduction.}

Let $I$ be an arbitrary instance of \pbvc.
The reduction maps $I$ to an instance $J$ of \pbsu{} that is defined as follows.
Let $G$, $V$, $E$, and $k$ be such that $I = (G, k)$ and $G = (V, E)$.
Let $n$ denote the cardinality of $V$.
Without loss of generality, 
we may assume 
$V = \seg{1}{n}$ 
and 
$k < n$ because $I$ is a yes-instance of \pbvc{} whenever $k \ge n$. 
Let $\left( y_e \right)_{e \in E}$, $\left( z_e \right)_{e \in E} \in V^E$ be such that $e = \left\{ y_e,  z_e \right\}$ for every $e \in E$.
Set 
 \begin{gather*}
 A = \{ \emptyset \} \cup E \, , \\
 s =  {(n + 4)}^3\,, \\
 R =  R_{n + 4}\,, \\
 X_\emptyset = (V - s - n) \cup  (R -  s) \cup (R + n) \cup (V + s + n)\,, \\
 X_e =  \left\{ z_e - n \right\} \cup R  \cup  \left\{ y_e + s  \right\} 
\end{gather*}
for each $e \in E$, and 
\begin{gather*}
J = \left( \inda{X}, \card{X_\emptyset} + k \right) \, . 
\end{gather*}
Clearly, $J$ is computable from  $I$ in polynomial time.

\paragraph{An instance of \pbaux.}
 Let us prove that $J$ is in fact an instance of \pbaux.
With the help of Lemmas~\ref{lem:golomb}.\ref{ppty:subset} and~\ref{lem:golomb}.\ref{ppty:card},
it is easy to see 
that $1 - s - n$ is the least element of $\bigcup_{a \in A} X_a$,
that $s + 2n$ is the greatest element of $\bigcup_{a \in A} X_a$, and 
that the cardinality of $X_\emptyset$ equals $4n + 8$. 
The latter property implies $\card{X_\emptyset} + k < 5 n + 8$.
Hence, 
the absolute value of every integer in 
$\left\{ \card{X_\emptyset} + k \right\} \cup \bigcup_{a \in A} X_a$ 
is at most $s + 2n$.
Now, remark that $s + 2n = f(\card{X_\emptyset}) \le f\left( \max_{a \in A} \card{X_a} \right)$. 

\paragraph{Correctness of the reduction.}
It remains to prove that 
$I$ is a yes-instance of \pbvc{} 
if, and only if,
$J$ is a yes-instance of \pbsu.

\begin{lemma} \label{lem:Xens}
For every $e \in E$,
it holds true that 
\begin{enumerate}
 \item \label{ppty:s} $(X_e -  s) \setminus  X_\emptyset = \left\{  y_e  \right\}$ and that
 \item \label{ppty:n} $(X_e +  n) \setminus  X_\emptyset =  \left\{  z_e  \right\}$.
\end{enumerate}
\end{lemma}

\begin{proof}
 We only prove Property~\ref{ppty:s} because Property~\ref{ppty:n} can be proven in the same way. 
Put $Y  =  \left\{ z_e - s - n \right\} \cup (R - s)$.
It is clear that 
$X_e - s  =  Y  \cup  \left\{ y_e  \right\}$  
and 
$Y  \subseteq X_\emptyset$. 
Therefore, we have 
\begin{equation} \label{eq:Xe-Y-y}
 X_e - s = \left\{ y_e  \right\} \setminus X_\emptyset \, .
\end{equation}
Moreover,  it follows from Lemma~\ref{lem:golomb}.\ref{ppty:subset} that 
\begin{itemize}
 \item the greatest element of ${(V - s - n)} \cup {(R - s)}$ equals $0$ and that
 \item the least element of ${(R + n)} \cup {(V + s + n)}$ equals $n + 1$.
\end{itemize}
Therefore, $X_\emptyset$ does not contain any element of $\seg{1}{n}$.
In particular, $y_e$ does not belong to $X_\emptyset$.
Combining the latter fact with Equation~\eqref{eq:Xe-Y-y}, 
we obtain  Property~\ref{ppty:s}.
\qed
\end{proof}

\begin{lemma} \label{lem:tpms}
For every $t \in \Zset$, 
$\card{(R + t) \setminus X_\emptyset} <  n$ implies $t \in \left\{ - s, + n \right\}$. 
\end{lemma}

\begin{proof}
Let us first bound from above the cardinality of $(R + t) \cap X_\emptyset$.
For each $\tau \in \Zset$, put
$
P_\tau = (R  + t) \cap (V + \tau)
$
and  
$
Q_\tau = (R  + t) \cap (R + \tau).
$
First, 
it follows from Lemma~\ref{lem:golomb}.\ref{ppty:distn} that 
$\card{P_\tau} \le 1$.
Second, 
$\tau \ne t$ implies $\card{Q_\tau} \le 1$ by Lemma~\ref{lem:golomb}.\ref{ppty:Golomb}.
And third, 
it holds that
$$(R + t) \cap X_\emptyset =  P_{- s - n} \cup Q_{-s} \cup Q_n \cup P_{s + n} \, .$$
Now, assume $t \notin \{ - s, + n \}$.
From the preceding three facts, 
we deduce that
$$
\card{(R + t) \cap X_\emptyset} 
\le 
\card{P_{-s - n}} + 
\card{Q_{-s}} + 
\card{Q_n} + 
\card{P_{s + n}} 
\le 
4 \,. 
$$
(In fact, 
it is not hard to see that $\card{(R + t) \cap X_\emptyset} \le 2$ holds: 
$t >  +n$ implies  $P_{- s - n} = Q_{-s} = \emptyset$,  
$- s <  t < +n$ implies $P_{- s - n}  = P_{s + n} = \emptyset$, and 
$t < - s$ implies $Q_n = P_{s + n} = \emptyset$.)
Since $\card{R + t} = n + 4$ by Lemma~\ref{lem:golomb}.\ref{ppty:card}, 
we finally get 
$$
\card{(R + t) \setminus X_\emptyset} = n + 4 - \card{(R + t) \cap X_\emptyset }  \ge  n  \, .
$$ 
\qed 
\end{proof}

\paragraph{(If).} 

Assume that $I$ is a yes-instance of \pbvc.
Then, there exists a vertex cover $C$ of $G$ with $\card{C} \le k$.
Put 
$F = \left\{ e \in E : y_e \in C \right\}$.
Set
$t_\emptyset = 0$, 
$t_e = -s$ for each $e \in F$, and 
$t_e = +n$ for each $e \in E \setminus F$.
On the one hand, 
it holds that 
\begin{equation} \label{eq:Xa-ta-X0}
\card{\bigcup_{a \in A} (X_a + t_a) } = \card{X_\emptyset} + \card{\bigcup_{e \in E} (X_e + t_e ) \setminus X_\emptyset}  
\end{equation}
because $t_\emptyset = 0$.
On the other hand,
it follows from Lemma~\ref{lem:Xens} that 
\begin{equation} \label{eq:Xe-te-ye-ze}
\bigcup_{e \in E} (X_e + t_e) \setminus X_\emptyset 
= 
\left\{ y_e : e \in F \right\} \cup \left\{ z_e : e \in E \setminus F \right\} \, .
\end{equation}
Since the right-hand side of Equation~\eqref{eq:Xe-te-ye-ze} is a subset of~$C$, 
we have
\begin{equation} \label{eq:Xe-te-X0-k}
\card{\bigcup_{e \in E} (X_e + t_e) \setminus X_\emptyset  } \le  k \, .
\end{equation}
We then get   
\begin{equation} \label{eq:J-yes}
 \card{\bigcup_{a \in A} (X_a + t_a) } \le \card{X_\emptyset} + k 
\end{equation}
by combining Equations~\eqref{eq:Xa-ta-X0} and~\eqref{eq:Xe-te-X0-k}.
Hence, $J$ is a yes-instance of \pbsu.

\paragraph{(Only if).} 
Assume that $J$ is a yes-instance of \pbsu.
Then, there exists $\inda{t} \in \Zset^A$ such that Equation~\eqref{eq:J-yes} holds.
Replacing 
$\inda{t}$ with $\left( t_a - t_\emptyset \right)_{a \in A}$ 
leaves the cardinality of $\bigcup_{a \in A} (X_a + t_a)$ unchanged; 
therefore, 
we may assume that $t_\emptyset = 0$;
in particular, Equation~\eqref{eq:Xa-ta-X0} holds.

Put 
$$
C = \bigcup_{e \in E} (X_e + t_e) \setminus X_\emptyset  \, .
$$
Combining Equations~\eqref{eq:Xa-ta-X0} and~\eqref{eq:J-yes}, 
we obtain Equation~\eqref{eq:Xe-te-X0-k}, or equivalently, $\card{C} \le k$.
Now, let us prove that $C$ is a vertex cover of~$G$.
Consider an arbitrary edge~$e \in E$.
Since we have
$$
(R + t_e) \setminus X_\emptyset 
 \subseteq 
(X_e + t_e) \setminus X_\emptyset \subseteq C \, ,
$$ 
it follows from Lemma~\ref{lem:tpms} that $t_e \in \{  -s,  +n \}$.
Consequently, Lemma~\ref{lem:Xens} ensures that some extremity of $e$ belongs to $(X_e + t_e) \setminus X_\emptyset$,
and this extremity is \emph{a fortiori} in~$C$.
Hence, $I$ is a yes-instance of \pbvc.
\qed
\end{proof}

A graph $G = (V, E)$ is called \emph{cubic} if for every vertex $v \in V$, 
the degree of $v$ in $G$ 
(\emph{i.e.}, the cardinality of $\left\{  w \in V : \{ v, w \} \in E \right\}$) 
equals~$3$. 
Let \pbmvc$3$ denote the restriction of \pbmvc{} to cubic graphs.
\pbmvc$3$  is $\APX$-complete under L-reduction \cite{alimonti00APX};
moreover, if \pbmvc$3$ is $\frac{100}{99}$-approximable in polynomial time then $\PNP$ \cite{chlebik06complexity}.

To prove that \pbmsu{} is ``strongly'' $\APX$-hard,
which is a better result than Theorem~\ref{th:MSU-NPC},
we show that \pbmvc$3$ L-reduces to a suitable restriction of \pbmsu.
In fact, we simply adapt the proof of Theorem~\ref{th:MSU-NPC}.

\begin{theorem} \label{th:MSU-APX}
There exists a real constant $\rho > 1$ such that 
if \pbmsu{} is $\rho$-approximable in pseudo-polynomial time then $\PNP$.
\end{theorem}

\begin{proof}
Let $f$ be as in the proof of Theorem~\ref{th:MSU-NPC} 
and 
let \pbmaux{} denote the restriction of \pbmsu{} to those instances $\inda{X}$ such that 
the absolute value of every integer in $\bigcup_{a \in  A} X_a$ is at most $f\left( \max_{a \in A} \card{X_a} \right)$.
We prove that \pbmaux{} is $\APX$-hard, 
which implies the theorem 
because 
every pseudo-polynomial-time approximation algorithm for \pbmsu{} 
is 
 a polynomial-time approximation algorithm for \pbmaux.
More precisely, we show that 
\pbmvc$3$ {L-reduces} \cite{ausiello03complexity} to \pbmaux.
We use the notation of the proof of Theorem~\ref{th:MSU-NPC}.

\paragraph{From a graph to an instance of \pbaux.}

Let $\tau$ denote the minimum cardinality of a vertex cover of~$G$.
Let $\upsilon$ denote the minimum cardinality of $\bigcup_{a \in A} (X_a + t_a)$ over all $\inda{t} \in \Zset^A$.
Clearly, 
$\inda{X}$ is computable from $G$ in polynomial time ($\inda{X}$ is independent of $k$),
$\inda{X}$ is an instance of \pbmaux, and 
$\upsilon = \card{X_\emptyset} + \tau = 4n + 8 + \tau$.

Now, assume that $G$ is cubic and $n \ge 24$.
The first assumption implies $3 \tau \ge \card{E} \ge n$.
It follows 
$$
4n + 8 
= \left( 4 + \frac{8}{n}  \right) n 
\le \left( 12  + \frac{24}{n}  \right) \tau
\le 13\tau \, ,
$$
and thus $\upsilon \le 14 \tau$.

\paragraph{From a solution of \pbaux{} to a vertex cover.}

Let $\inda{t} \in \Zset^A$. 
Put 
$$k = \card{\bigcup_{a \in A} (X_a + t_a)} - \card{X_\emptyset}\,.$$
There exists a  vertex cover $C$ of $G$ 
that satisfies 
$\card{C} \le  k$, 
 or equivalently, 
$$\card{C} - \tau \le  \card{\bigcup_{a \in A} (X_a + t_a)} - \upsilon \, . $$
Moreover, such a vertex cover is computable from  $G$ and $\inda{t}$ in polynomial time:
\begin{itemize}
 \item if $k  \ge n$ then set $C = V$ and 
 \item if $k  < n$ then set $C = \bigcup_{e \in E} (X_e + t_e - t_\emptyset) \setminus X_\emptyset$.
\end{itemize}

\paragraph{Conclusion.}

Let $\varepsilon$ be a positive real number.
If \pbmsu{} is ${(1 + \varepsilon)}$-approximable in pseudo-polynomial time 
then \pbmvc$3$ is ${(1 + 14\varepsilon)}$-approximable in polynomial time.
Therefore, if \pbmsu{} is $\frac{1387}{1386}$-approximable in pseudo-polynomial then $\PNP$.
\qed
\end{proof}

An immediate corollary of Theorem~\ref{th:MSU-APX} is that \pbmsu{} does not
admit any (pseudo-)polynomial time approximation scheme.


\section{Open questions}

The following three questions remain open:
Does there exist a constant $\rho > 1$ such that 
\pbmsu{} is $\rho$-approximable in (pseudo-)polynomial time?
Is \pbsu{} \emph{fixed-parameter tractable} \cite{flum06parameterized} with respect to parameter $\card{A}$?
Is \pbsu{} solvable in polynomial time for bounded $\max_{a \in A} \card{X_a}$?

%
%



\bibliographystyle{abbrv}

\bibliography{bibtex/group-literature}

\end{document}